\documentclass[runningheads]{llncs}
\usepackage{amsmath}
\usepackage{hyperref}
\usepackage[T1]{fontenc}
\usepackage{float}
\usepackage{xcolor}

\usepackage{amssymb}
\usepackage{graphicx}

\newcommand{\exampleend}{\hfill\(\triangleleft\)}

\begin{document}
\title{Facet-Defining Inequalities for the Angle-Based DC Optimal Transmission Switching Formulation}
\titlerunning{Facet-Defining Inequalities for the Angle-Based DC-OTS Formulation}
% If the paper title is too long for the running head, you can set
% an abbreviated paper title here
%\orcidID{0009-0009-1374-6833} \orcidID{0000-0002-4843-3564}
\author{Behnam Jabbari-Marand \and
Adolfo R. Escobedo}
\authorrunning{B. Jabbari-Marand et al.}
% First names are abbreviated in the running head.
% If there are more than two authors, 'et al.' is used.
%
\institute{Edward P. Fitts Department of Industrial and Systems Engineering, North Carolina State University, Raleigh, NC 27606, United States of America}
\maketitle              % typeset the header of the contribution
\begin{abstract}
The switching of transmission lines can significantly improve the economic and operational efficiency of power systems. The Direct-Current Optimal Transmission Switching (DC-OTS) problem provides a formal framework for minimizing power generation costs by reconfiguring the transmission network topology under a linearized power flow model. DC-OTS is typically formulated as a mixed-integer linear program that incorporates disjunctive constraints to capture the required relationships between certain variables via big-$M$ parameters. More specifically, these parameters represent upper bounds on voltage angle differences across non-operational transmission lines. In practice, overly conservative (and arbitrary) bounds tend to be used. The belief is that tightening these values requires the solution of the computationally intractable longest path problem. This work challenges that view through a novel polyhedral analysis of the angle-based DC-OTS formulation. We construct an extended formulation for the convex hull of an angle-based relaxation and derive facet-defining inequalities that tighten angle-difference bounds.
\keywords{Optimal transmission switching \and Mixed-integer linear programming \and Extended formulation \and Convex hull}
\end{abstract}\vspace{-22pt}
\section{Introduction}\label{sec1}
As renewable energy sources become more prevalent and grid infrastructure continues to age, the need for cost-effective planning and operational decisions is becoming increasingly critical. Transmission Expansion Planning (TEP) is used to make long-term upgrades \cite{garver1970transmission}, but its deployment is often limited by capital, regulatory, and siting constraints \cite{madrigal2012transmission}. A complementary strategy is topology control via transmission line switching, which has emerged as a systematic approach for improving the operational efficiency of existing infrastructure. This is motivated by the fact that a fixed topology is rarely optimal across all time periods or operating conditions \cite{hedman2011review}. This strategy is mathematically formulated as the Optimal Transmission Switching (OTS) problem, which determines an efficient network configuration by selectively switching transmission lines on or off \cite{fisher2008optimal}. Its benefits are well documented and include improved congestion management and controllability, along with reduced operating costs \cite{hedman2009optimal,escobedo2013topology,numan2023role}.

OTS extends Optimal Power Flow (OPF)—which minimizes generation costs subject to physical and operational constraints—by incorporating binary line-status variables that modify the network topology. The highest-fidelity OPF, Alternating Current OPF (ACOPF), is computationally challenging and offers only local optimality guarantees, as it is non-convex \cite{skolfield2022operations}. Introducing discrete line-switching decisions further exacerbates this intractability. Thus, various studies adopt the linearized Direct Current OPF (DCOPF) model (e.g., \cite{dey2022node,crozier2022feasible,pineda2024tight,taheri2025ac}), which simplifies the physics yet retains relevance to practice. DCOPF-based OTS (DC-OTS) introduces binary variables for reconfiguring the network topology, thereby transforming DCOPF from a linear program (LP) into a mixed-integer linear program (MILP). Although computationally more manageable than its ACOPF counterpart, DC-OTS remains NP-hard \cite{lehmann2014complexity}. The problem is typically modeled with ``big-$M$'' formulations (e.g., \cite{dey2022node},\cite{pineda2024tight},\cite{ramirez2022improving}), most notably the angle-based model of Fisher et al. \cite{fisher2008optimal}. This formulation is challenging to scale, largely due to its disjunctive constraints. We construct an extended formulation for an angle-based relaxation of this model and, using projection, derive valid inequalities (VIs) that describe its convex hull.
\subsection{Model and related work}
The angle-based DC-OTS formulation is presented in~\eqref{eq:tepobj}--\eqref{eq:cons8} \cite{fisher2008optimal}. The power network is modeled as a connected graph comprising a set of buses $\mathcal{B}$ (nodes) and a set of transmission lines $\mathcal{L}$ (edges). Each transmission line $(i, j) \in \mathcal{L}$, where $i, j \in \mathcal{B}$, has a key property known as its electrical reactance $x_{ij} > 0$. Each line allows bidirectional power flow, represented by the variable $f_{ij}$, and is subject to a \emph{thermal capacity} limit $\overline{f}_{ij}$~\eqref{eq:cons3}. The operational status of line $(i,j)$ is modeled by a binary variable $y_{ij}$, where $y_{ij} = 1$ if the line is \emph{active} (i.e., available for power transmission), and $y_{ij} = 0$ if the line is \emph{inactive} (i.e., disconnected from the network). Each bus may host a generator that produces active power $g_n$, subject to an upper limit $\overline{g}_n$\eqref{eq:cons6}. Generation at bus \(n\) incurs a linear cost of \(c_n\) per unit. The objective is to minimize the total cost of generation~\eqref{eq:tepobj}. At each bus $n$, Kirchhoff's Current Law (KCL) requires that the sum of incoming flows $f_{in}$ and generation $g_n$ equals the sum of outgoing flows $f_{ni}$ and demand, which is denoted by $d_n$~\eqref{eq:cons1}. The phase angle at each bus $n \in \mathcal{B}$ is represented by the variable \(\theta_n\). Constraint \eqref{eq:cons5} imposes Ohm's law on each line $(i,j)$: when the line is active, flow is proportional to the angle difference ($\theta_i-\theta_j$), implicitly enforcing Kirchhoff's Voltage Law (KVL). For inactive lines, a sufficiently large parameter $M_{ij}$ is used to ensure redundancy of this constraint.
\begin{subequations}\label{eq:DCOTS}
\begin{align}
 \min &\sum_{n \in \mathcal{B}}  c_n g_n && \label{eq:tepobj}\\[0pt]
\text{s.t.}&\sum_{(i, n) \in  \mathcal{L}}\hspace{-2pt} f_{i n} -\hspace{-2pt} \sum_{(n, i) \in \mathcal{L}} \hspace{-2pt} f_{n i} +g_n=d_n && \forall n \in \mathcal{B}  \label{eq:cons1}\\[0pt]
& -\overline{f}_{i j} y_{i j} \leq f_{i j} \leq \overline{f}_{i j} y_{i j} && \forall(i, j)\in \mathcal{L} \label{eq:cons3} \\
& -{M}_{ij}(1-y_{ij}) \leq f_{ij}x_{ij} - (\theta_i - \theta_j) \leq {M}_{ij}(1-y_{ij}) && \forall (i, j)\in \mathcal{L}  \label{eq:cons5}\\
& 0\leq g_n \leq \overline{g}_n, && \forall n \in \mathcal{B} \label{eq:cons6}\\[0pt]
& y_{i j} \in\{0,1\} && \forall(i, j)\in \mathcal{L} \label{eq:cons8}
\end{align}
\end{subequations}
Big-$M$ parameters ($M_{ij}$ in \eqref{eq:cons5}) impose an upper bound on the voltage angle difference between buses across inactive transmission lines. These constants must be large enough to preserve all integer-feasible solutions (since the optimal line statuses are not known \emph{a priori}), yet sufficiently tight to avoid weak LP relaxations that can substantially increase the computational burden of solving DC-OTS exactly \cite{dey2022node}.

To avoid large big-$M$ values, previous works adopt fixed angle-difference bounds—such as $2\pi$ or $1.2$ radians (e.g., \cite{dey2022node,hedman2010co,villumsen2012investment,peker2018two})—partly motivated by operational considerations. However, DCOPF models phase angles as real-valued (``non-periodic'') variables, and thus imposing fixed bounds can at times exclude feasible solutions~\cite{fattahi2018bound}. In addition, these bounds are typically specified only for adjacent buses and provide no guidance for the $\Theta(|\mathcal{B}|^2)$ non-adjacent pairs in the network. Even with fixed bounds, the computational complexity arising from the combinatorial nature of the solution space remains. This persistent complexity, coupled with the limited scalability of traditional algorithms, has led to an increased reliance on heuristic and data-driven approaches (e.g., \cite{crozier2022feasible,wu2013selection,johnson2021k,li2021data,pineda2024learning}).

In power systems with flexible network topologies, it is often assumed that the NP-hard longest path problem (LPP) needs to be solved to tighten the angle-difference bounds \cite{fattahi2018bound,binato2001new}. To avoid solving the LPP exactly, Aguilar-Moreno et al. \cite{aguilar2025graph} simplify the network and solve a relaxed problem to derive big-$M$ constants. However, the relaxation can further overestimate effective transfer paths, yielding big-$M$ values that remain excessively loose.

Alternatively, Kocuk et al. \cite{kocuk2016strong} and Dey et al. \cite{dey2022node} avoid big-$M$ issues via flow-based cuts. By characterizing the convex hulls of substructures excluding angle variables, they derive VIs that implicitly tighten angle-difference bounds and improve solution performance. Despite these advances, the polyhedral geometry of substructures explicitly including angle variables remains unexplored.

\medskip
\textbf{Key contributions.} This work addresses this gap by analyzing the polyhedral properties of a substructure of DC-OTS that includes angle variables, and it derives VIs that explicitly tighten bounds on their pairwise differences. We make two main contributions: (i) Define and reformulate a relaxation of the feasible set induced by this substructure; and (ii) conduct a polyhedral analysis of the relaxation and derive VIs that describe its convex hull.
 \section{Preliminaries}\label{sec2}
When a bus pair is connected through active lines—those fixed to be in service—a valid bound on their voltage-angle difference can be obtained efficiently from transmission line parameters and network topology. For an \emph{adjacent bus pair} $\langle i,j\rangle$ across an active line $(i,j)$, such a bound is obtained by combining the line capacity constraint~\eqref{eq:cons3} with Ohm's law~\eqref{eq:cons5}, resulting in the inequality
\begin{equation}
    |\theta_i - \theta_j| \leq  \overline{f}_{ij}x_{ij}. \label{eq:capacity-reactance}
\end{equation}
For ease of exposition, henceforth, we refer to the right-hand side of~\eqref{eq:capacity-reactance} as the \emph{weight} of line \((i,j)\) and denote it by the parameter \( w_{ij} \). Skolfield et al. \cite{skolfield2022derivation} derive analogous angle-difference bounds for \emph{any bus pair} $\langle m,n\rangle$ with $m\neq n$ (adjacent or non-adjacent) by summing weights along a path of interconnecting active lines. Let \(\rho_{mn} := \langle i_0:=m, i_1, \dots, i_{k-1}, i_k:=n \rangle\) denote such a path between buses \(m\) and \(n\). Summing the \emph{angle-difference inequalities}~\eqref{eq:capacity-reactance} along \(\rho_{mn}\) creates a telescoping effect on the left-hand side, which simplifies to the net difference \(|\theta_n - \theta_m|\), while the right-hand side accumulates the weights along the path:
\begin{subequations}\label{eq:telescop}
\begin{align} 
  \bigl\lvert \sum_{(i,j) \in \rho_{mn}} (\theta_{i} - \theta_{j}) \bigr\rvert \hspace{-2pt}\leq & \;
\underbrace{|\theta_{i_1} - \theta_{m}|}_{\substack{\leq w_{mi_1}}} +  
\underbrace{|\theta_{i_2} - \theta_{i_1}|}_{\substack{\leq w_{i_1i_2}}} +  
... +  
\underbrace{|\theta_{n} - \theta_{i_{k-1}}|}_{\substack{{\leq w_{i_{k-1}n}}}}&& \\[-4pt]
\Rightarrow |\theta_{n}-\theta_{m}|\hspace{-2pt} \leq &\;  w_{mi_1} + w_{i_1i_2}+... + w_{i_{k-1}n}=w(\rho_{mn}),  && \label{eq:angineq}
  \end{align}
  \end{subequations}
where \(w(\rho_{mn})\) denotes the cumulative weight along the path \(\rho_{mn}\). When buses \( m \) and \( n \) are connected by multiple paths, each path generates a corresponding \emph{path-based inequality} \eqref{eq:angineq} that provides an upper bound on \( |\theta_n - \theta_m| \). Among these, the shortest path yields the tightest bound. Hence, when the pair is connected, a valid and efficient upper bound on the angle difference can be obtained by solving a shortest path problem over the subgraph induced by the active lines~\cite{binato2001new}.

Efficient derivation of bus angle-difference bounds relies on the connectivity between the pairs through paths of active transmission lines. These results can be useful for TEP, where a fixed subgraph of active lines exists \cite{marandpolynomial}. However, the general case is more complex. In OTS, the connectivity between many bus pairs is unknown a priori because the topology is decision-dependent. Consequently, the number of potential paths connecting a pair can be prohibitively large, and bounds derived under fixed or partially known topologies may become invalid.

One workaround is to use a trivial bound for any flexible topology by summing the weights of all lines in the network \cite{tsamasphyrou2000transmission}. Although easy to compute, this bound ignores network structure and is typically very loose. To obtain tighter bounds, Binato et al. \cite{binato2001new} propose total enumeration of all simple paths connecting two buses to identify the longest one, ensuring validity under any switching decision. However, this approach is equivalent to solving the LPP, which is NP-hard~\cite{schrijver2003combinatorial}. Moreover, Fattahi et al. \cite{fattahi2018bound} show that, unless $P=NP$, even approximating the maximum feasible angle difference is intractable. Beyond these computational challenges, LPP-based bounds often yield big-$M$ values that substantially overstate realistic angle differences.

As noted earlier, previous polyhedral studies consider substructures that exclude angle variables and derive VIs that bound total power flows in certain cases. By imposing upper bounds on power flows, these VIs implicitly tighten angle-difference bounds—even though angle variables are not included in the inequalities—by leveraging their physical relationship (see~\eqref{eq:cons5}). In more detail, Kocuk et al. \cite{kocuk2016strong} introduce a \emph{cycle-based formulation} that enforces KVL directly on power flows around the cycles within a \textit{cycle basis}—a minimal set of linearly independent cycles that generates all graph cycles. Building on this formulation, they analyze the polyhedral structure of a cycle-based relaxation in $(f,y)$-space (including KVL, line capacity, and variable domain constraints), and introduce an extended formulation that describes its convex hull. From this analysis, they derive \emph{cycle-based VIs (CVIs)} and show that CVIs are facet-defining for the corresponding convex hull. For a cycle $C$, a CVI bounds the total flow over any line subset $S\subseteq C$; with $C$ oriented as a directed cycle, the inequality takes the form
\begin{equation}\label{eq:cyCVI}
\bigl\lvert \sum_{(i,j)\in S} f_{ij}x_{ij} \bigr\rvert \leq \Delta(S)(|C|-1) - \sum_{(i,j) \in S} [\Delta(S) - w_{ij}] y_{ij} - \Delta(S) \hspace{-4pt} \sum_{(i,j) \in C \setminus S} y_{ij},
\end{equation}
where \( w(C)\hspace{-1pt}:=\hspace{-1pt}\sum_{(i,j)\in C} w_{ij} \), $w(S)\hspace{-1pt}:=\hspace{-1pt}\sum_{(i,j)\in S} w_{ij}$, and \( \Delta(S)\hspace{-1pt}:=\hspace{-1pt}w(S)\hspace{-1pt} - w(C\hspace{-1pt} \setminus\hspace{-1pt} S) \). The inequality is nontrivial if $\Delta(S)\hspace{-2pt}>\hspace{-2pt}0$. On the other hand, the lines in $S$ need not be contiguous.
\begin{figure}[t]
        \centering
        \includegraphics[width=0.84\linewidth]{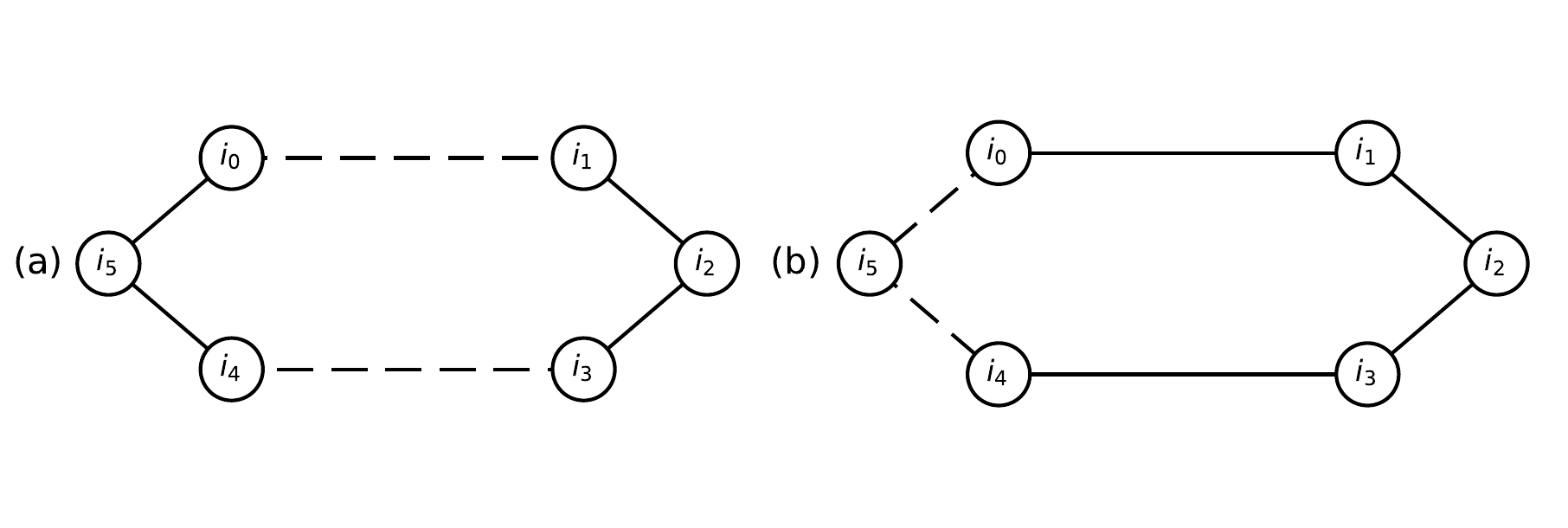}
        \vspace{-20pt}
        \caption{Cycle partitions used by each set of VIs: (a) Cycle-based VIs: solid edges form subset \(S\) whose total weight exceeds half of the cycle weight; dashed edges are \(C \setminus S\); (b) Path- and cycle-based VIs: solid edges are the longer path on the cycle by total weight (subset \(S\)); dashed edges are the shorter path (subset \(C \setminus S\)).}
        \label{fig:chart1}\vspace{-12pt}
\end{figure}

We draw a connection between the path-based inequalities \eqref{eq:telescop} and cycle-based inequalities \eqref{eq:cyCVI}. Unlike the cycle-based framework, which can use any subset $S\subseteq C$, the path-based approach requires the selected lines $S\subset C$ to be contiguous and form a simple path. This means that path-based inequalities can be generated from a cycle, but only when the latter is partitioned into two disjoint, complementary paths with the same terminal buses.
\begin{example}\label{ex:ex1} We assume unit weights on all lines. Figure~\ref{fig:chart1}a shows a selection \( S = \{(i_0, i_5), (i_4, i_5), (i_1, i_2), (i_2, i_3)\} \) (depicted by solid lines), with its complement \( C \setminus S = \{(i_0, i_1), (i_3, i_4)\} \) (dashed lines); both are discontiguous, and thus CVIs \eqref{eq:cyCVI} apply but path-based inequalities \eqref{eq:telescop} do not. By contrast, Figure~\ref{fig:chart1}b illustrates a valid partition for a path-based inequality on the same cycle: \( S = \{(i_0, i_1), (i_1, i_2), (i_2, i_3), (i_3, i_4)\} \) (solid lines) and \( C \setminus S = \{(i_0, i_5), (i_4, i_5)\} \) (dashed lines), where \( S \) forms the longer path and \( C \setminus S \) the shorter path (based on total weight) connecting bus pair \( \langle i_0,i_4\rangle \).\exampleend
\end{example}

CVIs tighten angle differences only indirectly, via bounds implied by power-flow relationships. This leaves an opportunity to introduce VIs that incorporate angle variables to achieve tighter relaxations. To capitalize on this, we analyze a relaxation of the angle-based DC-OTS model that explicitly includes angle differences and derive VIs describing its convex hull.
\section{Characterization of an angle-based relaxation}\label{sec3}
We introduce a new class of VIs that can tighten existing angle-difference bounds for both adjacent and non-adjacent bus pairs. These inequalities build on path-based structures (see \eqref{eq:telescop}) and are derived from an extended formulation of a relaxation of the DC-OTS model defined over a cycle \( C \). We show that, when projected, they become facet-defining in the original $(\theta,y)$-space over $C$.

Using constraints~\eqref{eq:cons3},~\eqref{eq:cons5}, and~\eqref{eq:cons8}, we define the cycle-based relaxation of the DC-OTS feasible region as
\begin{align}\label{eq:rc_def}
\mathcal{R}^{\theta,y,f}_C = \{(\theta,y,f): &f_{ij}x_{ij} - M_{ij}(1 - y_{ij}) \leq \theta_i - \theta_j \leq f_{ij}x_{ij}+ M_{ij}(1 - y_{ij}),\\\nonumber
-&\bar{f}_{ij} y_{ij} \leq f_{ij} \leq \bar{f}_{ij} y_{ij}, \quad y_{ij}\in\{0,1\}\quad \forall(i,j) \in C\},\nonumber
% &y_{ij}\in\{0,1\}, & (i,j)\in C\},
\end{align}
where $M_{ij}\geq \overline{f}_{ij}x_{ij}$ should be sufficiently large. Next, we replace the flow variables \( f_{ij} \) with line-capacity limits to derive implied bounds on voltage angle differences. This yields a relaxation of~\eqref{eq:rc_def}, which is denoted as $\mathcal{R}^{\theta,y}_C\supseteq \mathcal{R}^{\theta,y,f}_C$:
\begin{align}\nonumber
\mathcal{R}^{\theta,y}_C = \{(\theta,y): - &\overline{f}_{ij}x_{ij}y_{ij} - M_{ij}(1 - y_{ij}) \leq \theta_i - \theta_j \leq \overline{f}_{ij}x_{ij}y_{ij}+ M_{ij}(1 - y_{ij}),\\
&y_{ij}\in\{0,1\}\;\; \forall(i,j) \in C\}.
% &y_{ij}\in\{0,1\}, &(i,j)\in C
\end{align}
Furthermore, replacing $\overline{f}_{ij}x_{ij}$ with the line weight notation $w_{ij}$ and rewriting the double-sided inequality in absolute value form gives
\begin{equation}\label{eq:rthetax_def0}
\mathcal{R}^{\theta,y}_C = \{(\theta,y):| \theta_i - \theta_j |\leq w_{ij}y_{ij}+ M_{ij}(1 - y_{ij}), \quad y_{ij}\in\{0,1\} \quad \forall(i,j)\in C\}.
\end{equation}
Collecting the $y_{ij}$ terms and simplifying, set \eqref{eq:rthetax_def0} can be rewritten as
{\setlength{\abovedisplayskip}{4pt}%
 \setlength{\belowdisplayskip}{5pt}%
\begin{equation}\label{eq:rthetax_def}
\mathcal{R}^{\theta,y}_C = \{(\theta,y)\hspace{-2pt}: \hspace{-2pt}|\theta_i - \theta_j |\leq w_{ij}+ (M_{ij}-w_{ij})(1 - y_{ij}), \hspace{3pt} y_{ij}\in\{0,1\}\;\forall (i,j)\in C\}.
\end{equation}

Relaxation \eqref{eq:rthetax_def} consists of angle-difference inequalities only for adjacent buses. To enhance our structural characterization of $\mathcal{R}^{\theta,y}_C$, we write a more granular description by incorporating telescoping inequalities~\eqref{eq:telescop} along all possible paths connecting any pair of buses (including non-adjacent ones) in $C$. Let $\underline{\rho}_{mn}$ and $\overline{\rho}_{mn}$ denote, respectively, the \emph{shorter} and \emph{longer} paths (by total weight, i.e., $ w(\underline{\rho}_{mn}) < w(\overline{\rho}_{mn}) $) connecting $ \langle m, n \rangle$ in $ C $. We add path-based inequalities for all $\langle m, n \rangle\in C$ to yield a more detailed characterization of $\mathcal{R}^{\theta,y}_C$, denoted as $\mathcal{R}^{\theta,y}_{+_C}$:
\begin{align}\nonumber
\hspace{-0.2cm}\mathcal{R}^{\theta,y}_{+_C}=\{(\theta,y):&|\theta_n - \theta_m|\leq \hspace{-0.3cm}\sum_{(i,j)\in \underline{\rho}_{mn}} \hspace{-0.3cm} w_{ij} +\hspace{-0.4cm}\sum_{(i,j)\in \underline{\rho}_{mn}}\hspace{-0.25cm}  (M_{ij}-w_{ij})(1 - y_{ij}) \; \forall m,n \in C,\\[-2pt]\nonumber
&|\theta_n - \theta_m|\leq \hspace{-0.4cm}\sum_{(i,j)\in \overline{\rho}_{mn}} \hspace{-0.3cm} w_{ij} + \hspace{-0.4cm}\sum_{(i,j)\in \overline{\rho}_{mn}}\hspace{-0.25cm} (M_{ij}-w_{ij})(1 - y_{ij}) \; \forall m,n \in C,\\[-2pt]
&y_{ij}\in\{0,1\}\; \; \forall (i,j)\in C\}. \label{eq:rtheta_x_c_plus}
\end{align}
The inequalities in \eqref{eq:rtheta_x_c_plus} depend only on angle differences (as opposed to absolute angles) and, in DCOPF, only relative angles within each connected component are meaningful. Hence, we define $\delta\theta_{mn} := \theta_n - \theta_m$ for clarity of exposition, with $\delta\theta_{mn}$ representing the single degree of freedom in $\theta_n - \theta_m$, rewriting it as 
\begin{subequations}\label{eq:set00}
\begin{align}
      \hspace{-5pt}  \mathcal{R}^{\delta\theta,y}_{+_C}=\{(\delta\theta,y)\hspace{-3pt}:& \left|\delta\theta_{mn}\right|\leq   w(\underline{\rho}_{mn}) +\hspace{-11pt}\sum_{(i,j)\in \underline{\rho}_{mn}} \hspace{-8pt}(M_{ij}-w_{ij})(1 - y_{ij})   \;\forall m,n \in C,\label{eq:shorterset}\\[-2pt]
       & \left|\delta\theta_{mn}\right|\leq  w(\overline{\rho}_{mn}) + \hspace{-11pt}\sum_{(i,j)\in \overline{\rho}_{mn}}\hspace{-8pt}(M_{ij}-w_{ij})(1 - y_{ij}) \;\forall m,n \in C,\\[-2pt]\nonumber
       & y_{ij}\in\{0,1\} \; \;\forall(i,j) \in C\;\}.
\end{align}
\end{subequations}
The path-based inequalities in~\eqref{eq:set00} are tight when all lines along the corresponding path are active, but can become quite loose otherwise. For a bus pair $\langle m,n\rangle\in C$, if a path—say $\underline{\rho}_{mn}$—is inactive, then $\sum_{(i,j) \in \underline{\rho}_{mn}} (1 - y_{ij}) \ge 1$ (i.e., at least one line is inactive). In this case, the right-hand side of~\eqref{eq:shorterset} can become overly large because it sums multiple $M_{ij}$ values. This motivates incorporating a fallback bound via the inequality $\left|\delta\theta_{mn}\right| \leq M_{mn}$ for each bus pair $\langle m,n\rangle$, where $M_{mn}$ must be no smaller than the longest-path weight of the pair.

Next, to simplify notation and avoid solving LPP, we adopt a single big-$M$ for all pairs—the worst-case bound $M := \sum_{(i,j)\in \mathcal{L}} w_{ij}$ \cite{tsamasphyrou2000transmission}—and, without loss of generality, replace all coefficients $M_{ij}$ in \eqref{eq:set00} with $M$.

Finally, to simplify the characterization of this structure, note that set \eqref{eq:set00} is the intersection of the feasible sets induced by the path-based inequalities for all $\langle m,n\rangle \in C$. We isolate the set for each pair $\langle m,n\rangle$ and denote it by $\mathcal{R}^{\delta\theta,y}_{+_{\langle m,n \rangle}}$:
\begin{subequations}\label{eq:set}
\begin{align}
        \mathcal{R}^{\delta\theta,y}_{+_{\langle m,n \rangle}}=&\{\;(\delta\theta,y)\in \mathbb{R}^1 \times \{0,1\}^{|C|}:\notag \\[-2pt]
        & \left|\delta\theta_{mn}\right|\leq  w(\underline{\rho}_{mn}) + \sum_{(i,j)\in \underline{\rho}_{mn}} (M -w_{ij})(1 - y_{ij}),\label{eq:shorter}\\[-2pt]
       & \left|\delta\theta_{mn}\right|\leq w(\overline{\rho}_{mn}) +\sum_{(i,j)\in \overline{\rho}_{mn}} (M -w_{ij})(1 - y_{ij}),\label{eq:longer}  \\[-5pt]
       & \left|\delta\theta_{mn}\right|\leq M,\label{eq:validbound}\\
       &y_{ij}\in\{0,1\}\quad \forall(i,j) \in C\;\}.
\end{align}
\end{subequations}
\subsection{Derivation of valid inequalities}
We construct the extended formulation \eqref{eq:extendedform} to explicitly capture the disjunctive structure of $\mathcal{R}^{\delta\theta,y}_{+_{\langle m,n \rangle}}$. To that end, define binary variables \(\underline{z}_{mn}\) and \(\overline{z}_{mn}\) to indicate whether the respective \emph{path is active}.
That is, \(\underline{z}_{mn} = 1\) (resp. \(\overline{z}_{mn} = 1\)) if and only if all lines on \(\underline{\rho}_{mn}\) (resp. \(\overline{\rho}_{mn}\)) are active. In addition, define \(\overline{\zeta}_{mn}=\overline{z}_{mn}(1 - \underline{z}_{mn})\), which equals 1 if and only if \(\overline{\rho}_{mn}\) is active but \(\underline{\rho}_{mn}\) is not; this ensures that the longer path is considered only if the shorter is inactive. The angle-difference bound is then determined based on these indicators: if neither path is active, the bound on \(\lvert \delta\theta_{mn}\rvert\) defaults to \(M\); otherwise, it tightens according to the active path(s).
\begin{subequations}\label{eq:extendedform}
\begin{align}
  & 
  \left\{ 
  \begin{array}{l}
     \underline{z}_{mn}  \leq y_{ij} \quad \forall (i,j) \in \underline{\rho}_{mn}, \\
     \underline{z}_{mn} \geq \sum_{(i,j) \in \underline{\rho}_{mn}} y_{ij} -|\underline{\rho}_{mn}|+ 1,
  \end{array}
  \right.  \label{eq:lowlower0} \\
  &
  \left\{
  \begin{array}{l}
     \overline{z}_{mn}  \leq y_{ij} \quad \forall (i,j) \in \overline{\rho}_{mn}, \\
     \overline{z}_{mn} \geq \sum_{(i,j) \in \overline{\rho}_{mn}} y_{ij} -\left|\overline{\rho}_{mn}\right|+ 1,
  \end{array}
  \right. \label{eq:uplower0} \\
  & 
  \left\{
  \begin{array}{l}
     0\leq\overline{\zeta}_{mn} \leq \overline{z}_{mn}, \\
     0\leq\overline{\zeta}_{mn} \leq 1 - \underline{z}_{mn}, \\
     \overline{\zeta}_{mn} \geq \overline{z}_{mn} - \underline{z}_{mn},
  \end{array}
  \right. \label{eq:8a} \\
  &
  \left|\delta\theta_{mn}\right| \leq 
  w(\underline{\rho}_{mn})\,\underline{z}_{mn} \;+\;
  w(\overline{\rho}_{mn})\,\overline{\zeta}_{mn} \;+\;
  M (1-\underline{z}_{mn}-\overline{\zeta}_{mn}), \label{eq:convcomb} \\
  &
  y_{ij}\in [0,1]^{|C|}, \; \underline{z}_{mn}, \; \overline{z}_{mn} \in [0,1]. \label{eq:vardom}
\end{align}
\end{subequations}
To introduce Theorem~\ref{prop:theo01}, we recall two basic concepts. A formulation is \emph{locally ideal} if every vertex of its LP relaxation assigns integral values to the integer variables, and it is \emph{sharp} if the projection of its LP relaxation onto the original-variable space equals the convex hull of the mixed-integer set \cite{vielma2015mixed}.

\medskip
\begin{theorem}\label{prop:theo01}
Define polytope $\mathcal{P}_{\langle m,n \rangle} :=\{ (\delta\theta, y, \underline{z}, \overline{z}, \overline{\zeta}): \eqref{eq:extendedform}\}$. Linear system \eqref{eq:extendedform} is an extended formulation of
$\operatorname{conv}(\mathcal{R}^{\delta\theta,y}_{+_{\langle m,n\rangle}})$. Equivalently, $\mathcal{P}_{\langle m,n \rangle}$ is sharp with respect to $(\delta\theta,y)$—that is, $\operatorname{conv}(\mathcal{R}^{\delta\theta,y}_{+_{\langle m,n \rangle}})= \operatorname{proj}_{\delta\theta, y}\bigl(\mathcal{P}_{\langle m,n \rangle}\bigr)$.
\end{theorem}
\begin{proof}
\emph{Disjunctive Decomposition of $\mathcal{R}^{\delta\theta,y}_{+_{\langle m,n \rangle}}$.} 
Set \(\mathcal{R}^{\delta\theta,y}_{+_{\langle m,n \rangle}}\) can be partitioned according to three mutually exclusive and exhaustive cases: (i) the shorter path \(\underline{\rho}_{mn}\) is active, and \(\overline{\rho}_{mn}\) is either active or inactive; (ii) the shorter path is inactive, and the longer path \(\overline{\rho}_{mn}\) is active; and (iii) neither path is fully active. Each case corresponds to a disjunctive polyhedral region:
\begin{itemize}
   \item \textbf{Case i:} We have that \(\underline{z}_{mn} = 1\). Define the polytope
 \begin{subequations}\label{eq:disj1}
        \begin{align}
      \hspace{60pt}  \mathcal{P}^{(1,*)}_{\langle m,n \rangle} := \bigl\{ (\delta\theta, y): & \left|\delta\theta_{mn}\right| \leq w(\underline{\rho}_{mn}), \label{eq:shortactive0}\\[-1pt]
          \hspace{60pt}   & \sum_{(i,j) \in \underline{\rho}_{mn}} y_{ij} = |\underline{\rho}_{mn}|, \label{eq:shortactive1}\\[-1pt]
         \hspace{60pt}    & 0 \leq y_{ij} \leq 1, \quad \forall (i,j) \in C \bigr\}.  \label{eq:yvardom}
        \end{align}
    \end{subequations}
    \item \textbf{Case ii:} We have \(\overline{\zeta}_{mn}=1\) (equivalently, \((\underline{z}_{mn}, \overline{z}_{mn}) = (0,1)\)). Define the polytope $\mathcal{P}^{(0,1)}_{\langle m,n \rangle} :=\bigl\{ (\delta\theta, y) : \eqref{eq:yvardom},$
    \begin{subequations}\label{eq:disj2}
\begin{align}
   & \hspace{4.13cm} \left|\delta\theta_{mn}\right| \leq w(\overline{\rho}_{mn}),\\[-1pt]
   & \hspace{4.1cm}\sum_{(i,j) \in \overline{\rho}_{mn}} y_{ij} = \left|\overline{\rho}_{mn}\right|, \\[-1pt]
   & \hspace{4.1cm}\sum_{(i,j) \in \underline{\rho}_{mn}} y_{ij} \leq|\underline{\rho}_{mn}|- 1\bigr\}. \label{eq:shortinactive}
\end{align}
\end{subequations}
    \item \textbf{Case iii:} We have that $\underline{z}_{mn} = 0$ and $\overline{\zeta}_{mn}=\overline{z}_{mn}(1-\underline{z}_{mn}) = 0$. Define the polytope $\mathcal{P}^{(0,0)}_{\langle m,n \rangle} := \bigl\{ (\delta\theta, y):\eqref{eq:yvardom},\eqref{eq:shortinactive},$
    \begin{subequations}\label{eq:disj2}
\begin{align}
   & \hspace{5.5cm}\left|\delta\theta_{mn}\right| \leq M,\\[-1pt]
   & \hspace{5.4cm}\sum_{(i,j) \in \overline{\rho}_{mn}} y_{ij} \leq|\overline{\rho}_{mn}|- 1\bigr\}. \label{eq:nonactive}
\end{align}
\end{subequations}
\end{itemize}
Polytopes $\mathcal{P}^{(1,*)}_{\langle m,n \rangle}$, $\mathcal{P}^{(0,1)}_{\langle m,n \rangle} $, and $\mathcal{P}^{(0,0)}_{\langle m,n \rangle}$ are mutually exclusive and collectively exhaustive. Therefore, every integer-feasible point of $\mathcal{R}^{\delta\theta,y}_{+_{\langle m,n \rangle}}$ lies in exactly one of these regions, i.e., $\mathcal{R}^{\delta\theta,y}_{+_{\langle m,n \rangle}} = \mathcal{P}^{(1,*)}_{\langle m,n \rangle} \cup \mathcal{P}^{(0,1)}_{\langle m,n \rangle} \cup \mathcal{P}^{(0,0)}_{\langle m,n \rangle}$, with the union being disjoint over integer feasible solutions. Consequently, the convex hull is expressed as $\operatorname{conv}(\mathcal{R}^{\delta\theta,y}_{+_{\langle m,n \rangle}}) = \operatorname{conv} \left( \mathcal{P}^{(1,*)}_{\langle m,n \rangle} \cup \mathcal{P}^{(0,1)}_{\langle m,n \rangle} \cup \mathcal{P}^{(0,0)}_{\langle m,n \rangle} \right).$

\emph{Extended Formulation.} Linear system \eqref{eq:extendedform} describes $\operatorname{conv}(\mathcal{R}^{\delta\theta,y}_{+_{\langle m,n \rangle}})$ in an extended space by (i) linking path-status indicators to the corresponding line-status variables, (ii) enforcing mutual exclusivity so that exactly one polytope is active at any extreme point, and (iii) expressing the angle-difference bound as a convex combination of the bounds corresponding to each active polytope.

  The three disjuncts differ only in (i) which path is active and (ii) which angle-difference bound applies. This shared structure enables a compact lifting using path-status indicators, rather than duplicating variables for each disjunct. Constraints \eqref{eq:lowlower0} link $\underline{z}_{mn}$ to the activity of the lines in $\underline{\rho}_{mn}$, and \eqref{eq:uplower0} similarly link $\overline{z}_{mn}$ to the lines in $\overline{\rho}_{mn}$. With bounds \(0 \le y,\underline{z}_{mn},\overline{z}_{mn} \le 1\), this gives an \emph{ideal} encoding of the conjunctions, implying integrality of all LP extreme points in \((y,\underline{z},\overline{z})\).

The variable $\overline{\zeta}_{mn}$ encodes the product $\overline{z}_{mn}(1 - \underline{z}_{mn})$, ensuring that the longer path is only selected when the shorter path is inactive. Constraints \eqref{eq:8a}, together with $\underline{z}_{mn},\overline{z}_{mn}\in[0,1]$, are the McCormick inequalities and describe the convex hull of the graph $\{(\underline{z},\overline{z},\overline{\zeta}):\overline{\zeta}_{mn}=\overline{z}_{mn}(1-\underline{z}_{mn})\}$ over $[0,1]^2$. With $\underline{z}_{mn}$ and $\overline{z}_{mn}$ integral at LP extreme points, the McCormick system reduces to the equality $\overline{\zeta}_{mn}=\overline{z}_{mn}(1-\underline{z}_{mn})$. Hence, at LP extreme points $\overline{\zeta}_{mn}\in\{0,1\}$.

In summary, at LP extreme points, $\underline{z}_{mn}$, $\overline{\zeta}_{mn}$, $1 - \underline{z}_{mn} - \overline{\zeta}_{mn}$ are in $\{0,1\}$ and sum to $1$, and hence exactly one of $\mathcal{P}^{(1,*)}$, $\mathcal{P}^{(0,1)}$, $\mathcal{P}^{(0,0)}$ is active, preserving the original disjunctive structure. Given mutual exclusivity, the angle-difference bound in the lifted space equals a convex combination of the per-case bounds, yielding \eqref{eq:convcomb}. As established above, \eqref{eq:extendedform} has LP extreme points with $(y,\underline{z},\overline{z},\overline{\zeta})$ integral—hence it is \emph{locally ideal}. Locally ideal formulations are sharp \cite{vielma2015mixed} and, hence, projecting $\mathcal{P}_{\langle m,n \rangle}$ onto $(\delta\theta,y)$ yields $\operatorname{conv}(\mathcal{R}^{\delta\theta,y}_{+_{\langle m,n \rangle}})$. \hfill $\square$
\end{proof}
Theorem~\ref{prop:theo01} shows that \( \operatorname{conv}(\mathcal{R}^{\delta\theta,y}_{+_{\langle m,n \rangle}} )\) is obtained by projecting the lifted polytope \( \mathcal{P}_{\langle m,n \rangle} \). Theorem~\ref{theo:theo2} projects \eqref{eq:extendedform} to derive cycle-induced path-based VIs (C-PVIs) that describe this convex hull in the original variables \( (\delta\theta, y) \).

\medskip
\begin{theorem} \label{theo:theo2}
For conciseness, define $\Delta(\rho):=w(\overline{\rho}_{mn}) - w(\underline{\rho}_{mn})$ and
$\Delta(M):= M - w(\overline{\rho}_{mn})$. The polytope \(\operatorname{conv}(\mathcal{R}^{\delta\theta,y}_{+_{\langle m,n \rangle}}) \) can be described by \eqref{eq:yvardom} and the C-PVI given by
\begin{equation}\label{eq:fori3}
\left|\delta\theta_{mn}\right| \le w(\underline{\rho}_{mn}) 
+ \Bigl(|\underline{\rho}_{mn}| - \hspace{-9pt}\sum_{(i,j)\in \underline{\rho}_{mn}}\hspace{-6pt} y_{ij} \Bigr) \Delta(\rho)+ \Bigl(\left|\overline{\rho}_{mn}\right| - \hspace{-9pt}\sum_{(i,j)\in \overline{\rho}_{mn}} \hspace{-6pt} y_{ij} \Bigr) \Delta(M).
\end{equation}
\end{theorem}
\begin{proof}
Theorem~\ref{prop:theo01} establishes that \( \operatorname{conv}(\mathcal{R}^{\delta\theta,y}_{+_{\langle m,n \rangle}}) = \operatorname{proj}_{\delta\theta,y}(\mathcal{P}_{\langle m,n \rangle}) \). It remains to show that inequalities \eqref{eq:yvardom} and \eqref{eq:fori3} fully characterize this projection.

\medskip
The angle-difference inequality in~\eqref{eq:convcomb} can be equivalently written as
\begin{equation}\label{eq:nlpvi}
    \left|\delta\theta_{mn}\right| \leq M-\underline{z}_{mn} \bigl(M-w(\underline{\rho}_{mn})\bigr) - \overline{\zeta}_{mn} \bigl(M-w(\overline{\rho}_{mn})\bigr).
\end{equation}
Since the coefficient of $\overline{\zeta}_{mn}$ in~\eqref{eq:nlpvi} is non-positive, the tightest bound in the projection is obtained by replacing it with its lower bound. Applying Fourier–Motzkin elimination and using the bound $\overline{\zeta}_{mn} \geq \overline{z}_{mn} - \underline{z}_{mn}$ from~\eqref{eq:8a} yields
 \begin{equation} \label{eq:fori1}
\left|\delta\theta_{mn}\right|
\;\le\;
M
- \underline{z}_{mn} \bigl( w(\overline{\rho}_{mn}) - w(\underline{\rho}_{mn}) \bigr)
- \overline{z}_{mn} \bigl( M - w(\overline{\rho}_{mn}) \bigr).
\end{equation}
Because the extended formulation is locally ideal, $(\underline{z}_{mn},\overline{z}_{mn})$ take integral values at all LP extreme points. Consequently, the projection of $\mathcal{P}_{\langle m,n \rangle}$ onto the space $(\delta\theta, y, \underline{z}, \overline{z})$ preserves the convex hull and is given by
\begin{equation}
    \hspace{-4pt} \operatorname{proj}_{\delta\theta,y,\underline z,\overline z}\bigl(\hspace{-1pt} \mathcal{P}_{\langle m,n\rangle}\hspace{-1pt} \bigr)\hspace{-2pt} =\hspace{-2pt} \bigl\{\hspace{-2pt} (\delta\theta, y, \underline{z}, \overline{z})\hspace{-2pt}  :\hspace{-2pt}\eqref{eq:lowlower0},\hspace{-1pt} \eqref{eq:uplower0},\hspace{-1pt} \eqref{eq:yvardom},\hspace{-1pt} \eqref{eq:fori1},\underline z_{mn}, \overline z_{mn}\hspace{-2pt} \in \hspace{-2pt} [0,1] \hspace{-1pt}\bigr\}.
\end{equation}
Substituting the lower bounds of $(\underline{z}_{mn},\overline{z}_{mn})$ from~\eqref{eq:lowlower0},\eqref{eq:uplower0} in \eqref{eq:fori1} then gives

\begin{equation}\label{eq:fori2}
\left|\delta\theta_{mn}\right|\le M 
- \Bigl( \hspace{-0.5pt}\sum_{(i,j)\in \underline{\rho}_{mn}}\hspace{-7pt} y_{ij} -|\underline{\rho}_{mn}|+ 1 \Bigr) \Delta(\rho)- \Bigl( \sum_{(i,j)\in \overline{\rho}_{mn}}\hspace{-7pt} y_{ij} -\left|\overline{\rho}_{mn}\right|+ 1 \Bigr) \Delta(M).
\end{equation}
Rearranging~\eqref{eq:fori2} yields C-PVI~\eqref{eq:fori3}, completing the projection. \hfill $\square$
\end{proof}

Having shown that C-PVIs \eqref{eq:fori3} describe \(\operatorname{conv}(\mathcal{R}^{\delta\theta,y}_{+_{\langle m,n \rangle}})\), we now proceed to a more detailed structural analysis to demonstrate that they are facet-defining, establishing their necessity in the complete description of \(\operatorname{conv}(\mathcal{R}^{\delta\theta,y}_{+_{\langle m,n \rangle}})\). We first verify that the convex hull is full-dimensional.

\begin{proposition}\label{prop:fulldim}
The polytope \(\operatorname{conv}(\mathcal{R}^{\delta\theta,y}_{+_{\langle m,n\rangle}})\), as defined below, is full-dimensional in \(\mathbb{R}^{|C|+1}\).
\begin{subequations}\label{eq:set}
\begin{align}
        \operatorname{conv}(\mathcal{R}^{\delta\theta,y}_{+_{\langle m,n \rangle}})=&\{\;(\delta\theta,y)\in \mathbb{R}^1 \times [0,1]^{|C|}:\notag \\
        & \left|\delta\theta_{mn}\right|\leq  w(\underline{\rho}_{mn}) + \sum_{(i,j)\in \underline{\rho}_{mn}} (M -w_{ij})(1 - y_{ij}),\label{eq:shorter2}\\
       & \left|\delta\theta_{mn}\right|\leq w(\overline{\rho}_{mn}) +\sum_{(i,j)\in \overline{\rho}_{mn}} (M -w_{ij})(1 - y_{ij}),\label{eq:longer2}  \\
       & \left|\delta\theta_{mn}\right|\leq M,\label{eq:validbound2}\\
       &0 \leq y_{ij}\leq 1\quad \forall(i,j) \in C\;\}.
\end{align}
\end{subequations}
\end{proposition}
\begin{proof}
\noindent\emph{Interior Point Construction.} Consider the point \((\delta\theta_{mn}^*, y^*) := (0, \tfrac{1}{2} \mathbf{1}_{|C|})\), where \(\mathbf{1}_{|C|}\) is the $|C|-$vector of ones. This point, with \( y^*_{ij} =\frac{1}{2}\) for every line \((i,j) \in C\), lies strictly in the interior of the bounds $0< y_{ij}< 1$. Moreover, since \(M > w(\overline{\rho}_{mn}) > w(\underline{\rho}_{mn})>0\), the right-hand sides of inequalities \eqref{eq:shorter2}-\eqref{eq:validbound2} are strictly positive at \(y^*\), ensuring strict feasibility.

\medskip
\noindent\emph{Perturbation Argument.}  
 Let \(\varepsilon > 0\), and let \(d \in \mathbb{R}^{|C|}\) satisfy \(\|d\|_\infty < \eta\) for some $\eta>0$. Assume that both $\varepsilon$ and $\eta$ are chosen to be sufficiently small. Define the perturbed point as $(\delta\theta_{mn}^\varepsilon, y^\eta) := \left( \varepsilon, \tfrac{1}{2} \mathbf{1} + d \right)$.
For sufficiently small \(\varepsilon\) and \(\eta\), \(y^\eta\) lies strictly within the open box \((0,1)^{|C|}\), and by continuity, all inequality constraints remain strictly satisfied as well. Therefore, an open neighborhood around \((\delta\theta_{mn}^*, y^*)\) lies entirely within the polytope. It follows that \(\operatorname{conv}(\mathcal{R}^{\delta\theta,y}_{+_{\langle m,n \rangle}})\) has nonempty interior in \(\mathbb{R}^{|C|+1}\) and is therefore full-dimensional. \hfill $\square$
\end{proof}

\smallskip
\begin{theorem}\label{theo:theo3}
C-PVIs \eqref{eq:fori3} are facet-defining for
$\operatorname{conv}\bigl(\mathcal{R}^{\delta\theta ,y}_{+_{\langle m,n\rangle}}\bigr)$.
\end{theorem}
\begin{proof}
 To show that C-PVIs are facet-defining, it suffices to construct \(|C|+1\) affinely independent integer points that satisfy the VIs at equality. By symmetry, we consider the case $\delta\theta_{mn}\geq 0$.
 
\medskip
\noindent
\emph{Construction of Points.}
Let $\mathbf{1}_{\underline{\rho}_{mn}}$ denote the $|\underline{\rho}_{mn}|$‑vector of ones (and similarly for $\overline{\rho}_{mn}$). Define the following $|C|+1$ points, represented as the ordered triple $(\delta\theta_{mn}, [y_{ij}]_{(i,j) \in \underline{\rho}_{mn}},[y_{ij}]_{(i,j) \in \overline{\rho}_{mn}})$:
\begin{itemize}
    \item \textbf{Base point:}
        $P_0:=\bigl(w(\underline{\rho}_{mn}),
              \mathbf{1}_{\underline{\rho}_{mn}},
              \mathbf{1}_{\overline{\rho}_{mn}}\bigr)= (\underbrace{w(\underline{\rho}_{mn})}_{\substack{\delta\theta_{mn}}}, \underbrace{1,1,\dots,1}_{\substack{(i,j)\in \underline{\rho}_{mn}}}, \underbrace{1,1,\dots,1}_{\substack{(i,j)\in\overline{\rho}_{mn}}}),$
which satisfies \eqref{eq:fori3} at equality, since the right-hand side reduces to $w(\underline{\rho}_{mn})$.\medskip
    \item \textbf{Points $P_1 \textbf{ to } P_{|\underline{\rho}|}$:}
        For each $k\in\{1,\dots,|\underline{\rho}_{mn}|\}$ define
            \[
           P_k := \hspace{-3pt}
            \bigl(w(\overline{\rho}_{mn}),
                  \mathbf{1}_{\underline{\rho}_{mn}}-e^{\underline{k}}_{|\underline{\rho}_{mn}|},
                  \mathbf{1}_{\overline{\rho}_{mn}}\bigr)\hspace{-3pt} \xrightarrow{k=1} \hspace{-3pt}\\ \nonumber
                  P_1:=(w(\overline{\rho}_{mn}), 
                \underbrace{0,1,\dots,1}_{\substack{(i,j)\in \underline{\rho}_{mn}}},
                \underbrace{1,\dots,1}_{\substack{(i,j)\in \overline{\rho}_{mn}}}),
            \]
            where $e^{\underline{k}}_{|\underline{\rho}_{mn}|}$ denotes the standard basis vector in $\mathbb{R}^{|\underline{\rho}_{mn}|}$, corresponding to the $k$‑th line in $\underline{\rho}_{mn}$. Each point $P_k$ satisfies \eqref{eq:fori3} at equality.\medskip
    \item \textbf{Points $P_{|\underline{\rho}|+1}\textbf{ to }P_{|\underline{\rho}|+|\overline{\rho}|}$:}
    % \textbf{One line opened on $\underline{\rho}_{mn}$ and $\overline{\rho}_{mn}$.}
        Fix the first line of $\underline{\rho}_{mn}$ as inactive and, for each $\ell\in\{1,\dots,\left|\overline{\rho}_{mn}\right|\}$, define $P_{|\underline{\rho}_{mn}|+\ell}:=
        \bigl(M,
              \mathbf{1}_{\underline{\rho}_{mn}} \hspace{-2pt}-e^{\underline{1}}_{|\underline{\rho}_{mn}|},
              \mathbf{1}_{\overline{\rho}_{mn}} \hspace{-2pt}-e^{\overline{\ell}}_{|\overline{\rho}_{mn}|}\bigr)$, where $e^{\overline{\ell}}_{|\overline{\rho}_{mn}|}$ denotes the standard basis vector in $\mathbb{R}^{\left|\overline{\rho}_{mn}\right|}$, associated with the $\ell$‑th element of $\overline{\rho}_{mn}$. Each point $P_{|\underline{\rho}_{mn}|+\ell}$ satisfies \eqref{eq:fori3} at equality. For $\ell=1$, the corresponding point is $P_{|\underline{\rho}_{mn}|+1} \hspace{-2pt}:=  \hspace{-2pt} (M,
            \underbrace{0,1,\dots,1}_{\substack{(i,j)\in \underline{\rho}_{mn}}},
            \underbrace{0,1,\dots,1}_{\substack{(i,j)\in \overline{\rho}_{mn}}})$.
\end{itemize}

\noindent
\emph{Affine Independence.}
Let $\{v_k := P_k - P_0\}_{k=1}^{|C|}$ denote the set of vectors obtained by translating each point $P_k$ relative to $P_0$. We will show that these vectors are linearly independent, which establishes that $\{P_k\}^{|C|}_{k=0}$ are affinely independent. For $k \in \{1,\ldots,|\underline{\rho}_{mn}|\}$, $v_k = (\Delta(\rho), -e^{\underline{k}}_{|\underline{\rho}_{mn}|}, \mathbf{0}_{\left|\overline{\rho}_{mn}\right|}),$ and for $k = |\underline{\rho}_{mn}|+\ell$, where $\ell \in \{1,\ldots,\left|\overline{\rho}_{mn}\right|\}$, we have $v_k = (\Delta(M), -e^{\underline{1}}_{|\underline{\rho}_{mn}|}, -e^{\overline{\ell}}_{\left|\overline{\rho}_{mn}\right|})$.

\medskip
To establish linear independence, assume scalars $\lambda_1,\ldots,\lambda_{|C|}\in \mathbb{R}$ satisfy $\sum_{k=1}^{|C|} \lambda_k v_k = \mathbf{0}$.
Expanding the linear combination yields the system
\begin{equation}\label{eq:linind}
  \hspace{-5pt} \sum_{k=1}^{|\underline{\rho}_{mn}|} \hspace{-3pt}\lambda_k \bigl(\hspace{-1pt} \Delta(\rho),\hspace{-2pt} -e^{\underline{k}}_{|\underline{\rho}_{mn}|}, \mathbf{0}_{\left|\overline{\rho}_{mn}\right|}\bigr) \hspace{-1pt}+\hspace{-5pt} \sum_{\ell=1}^{\left|\overline{\rho}_{mn}\right|} \hspace{-3pt}\lambda_{|\underline{\rho}_{mn}|+\ell} \bigl(\hspace{-1pt} \Delta(M), \hspace{-2pt}-e^{\underline{1}}_{|\underline{\rho}_{mn}|} , \hspace{-2pt}-e^{\overline{\ell}}_{\left|\overline{\rho}_{mn}\right|} \bigr)\hspace{-3pt}= \hspace{-3pt}\mathbf{0}_{|C|+1}.
\end{equation}
Then, isolating the third block components of each vector $v_k$ in~(\ref{eq:linind}) gives
\begin{equation}\label{eq:thirdblock}
     \sum_{k=1}^{|\underline{\rho}_{mn}|} \lambda_k \cdot \mathbf{0}_{\left|\overline{\rho}_{mn}\right|} - \sum_{\ell=1}^{\left|\overline{\rho}_{mn}\right|} \lambda_{|\underline{\rho}_{mn}|+\ell} \cdot e^{\overline{\ell}}_{\left|\overline{\rho}_{mn}\right|}  = - \sum_{\ell=1}^{\left|\overline{\rho}_{mn}\right|} \lambda_{|\underline{\rho}_{mn}|+\ell} \cdot e^{\overline{\ell}}_{\left|\overline{\rho}_{mn}\right|} = \mathbf{0}_{\left|\overline{\rho}_{mn}\right|}.
\end{equation}
Equation \eqref{eq:thirdblock} holds if and only if $\lambda_{|\underline{\rho}_{mn}|+\ell}=0$ for all $\ell\in\{1,...,\left|\overline{\rho}_{mn}\right|\}$, as the standard basis vectors $e^{\overline{\ell}}_{\left|\overline{\rho}_{mn}\right|}$ are linearly independent. With these terms eliminated, the second block in \eqref{eq:linind} reduces to
\begin{equation}\label{eq:secondblock}
   - \sum_{k=1}^{|\underline{\rho}_{mn}|} \lambda_k \cdot  e^{\underline{k}}_{|\underline{\rho}_{mn}|} - \sum_{\ell=1}^{|\overline{\rho}_{mn}|} \lambda_{|\underline{\rho}_{mn}|+\ell} \cdot e^{\underline{1}}_{|\underline{\rho}_{mn}|}  =  - \sum_{k=1}^{|\underline{\rho}_{mn}|} \lambda_k \cdot  e^{\underline{k}}_{|\underline{\rho}_{mn}|} = \mathbf{0}_{|\underline{\rho}_{mn}|}.
\end{equation}
Similarly, equation \eqref{eq:secondblock} holds if and only if $\lambda_k=0$ for all $k \in \{1,\ldots,|\underline{\rho}_{mn}|\}$, since the vectors $e^{\underline{k}}_{|\underline{\rho}_{mn}|}$ are standard basis vectors and thus linearly independent.

Therefore, $\lambda_k=0$ for $k \in \{1,\ldots,|C|\}$, establishing that the vectors $\{v_k\}^{|C|}_{k=1}$ are linearly independent. It follows that the points $\{P_k\}^{|C|}_{k=0}$ are affinely independent. Thus, C-PVIs \eqref{eq:fori3} are facet-defining for $\operatorname{conv}(\mathcal{R}^{\delta\theta,y}_{+_{\langle m,n\rangle}})$. \hfill $\square$
\end{proof}
\subsection{Separation}
C-PVIs can be separated efficiently by enumerating bus pairs within each cycle and applying a computationally inexpensive screening test. Due to space limitations, we only sketch the procedure. For each cycle $C$ and each bus $m\in C$, we incrementally extend a candidate shorter path $\underline{\rho}_{mn}$ by advancing $n$ along the cycle. We form the complementary path $\overline{\rho}_{mn}$ that closes the cycle and evaluate the cut violation only when the LP solution $(\delta\theta^{*},y^*)$ satisfies a necessary screening condition (e.g., $\lvert \delta\theta^{*}_{mn}\rvert > w(\underline{\rho}_{mn})$). The traversal from $m$ stops once the accumulated path weight exceeds $w(C)/2$ (i.e., $w(\underline{\rho}_{mn})>w(C)/2$). To keep the separation routine tractable, we apply this routine only to a subset of promising cycles identified from the LP solution. Implementation and computational evaluation are left to future work.
\section{Conclusion}\label{sec5}
This work analyzes the polyhedral structure of an angle-based relaxation of DC-OTS and introduces facet-defining cycle-induced path-based valid inequalities (C-PVIs). Unlike prior approaches, the proposed inequalities directly tighten angle-difference bounds across the network, and they can be calculated in polynomial time. When the selected subsets of lines in the cycle are contiguous, C-PVIs exhibit structural similarities with cycle-based VIs (CVIs) of Kocuk et al. \cite{kocuk2016strong} that motivate a comparison; however, a direct comparison is not possible because CVIs are stated in $(f,y)$-space. Future work will map CVIs into $(\theta,y)$-space to compare the resulting implied angle-difference bounds, extend the framework to more general substructures, benchmark CVIs and C-PVIs, and evaluate their impact on large-scale instances.

\begin{credits}
% \subsubsection{\ackname} A bold run-in heading in small font size at the end of the paper is
% used for general acknowledgments, for example: This study was funded
% by X (grant number Y).

\subsubsection{\discintname}
\small The authors have no competing interests to declare that are relevant to the content of this article.
\end{credits}

\end{document}